\documentclass[a4paper,USenglish]{article}

\usepackage{amssymb}
\usepackage{amsthm}
\usepackage{amsmath,mathtools}
\usepackage{enumitem}
\usepackage[utf8]{inputenc}
\usepackage[numbers]{natbib} %
\usepackage{todonotes}
\usepackage{xspace}
\usepackage{a4wide}

\usepackage[vlined,linesnumbered,ruled]{algorithm2e}
\SetEndCharOfAlgoLine{}
\SetCommentSty{textrm}

\usepackage[pagebackref,pdfdisplaydoctitle,menucolor=orange!40!black,filecolor=magenta!40!black,urlcolor=blue!40!black,linkcolor=red!40!black,citecolor=green!40!black,colorlinks]{hyperref}
\usepackage[capitalize,sort&compress,nameinlink,noabbrev]{cleveref}
\crefname{appendix}{}{}
\usepackage{tikz}
\usetikzlibrary{calc,arrows.meta}

\usepackage{rotating}
\usepackage{booktabs}
\usepackage{multirow}
\usepackage{tabularx,longtable}
\usepackage{lscape}

\usepackage{array}
\makeatletter
\newcommand{\thickhline}{%
    \noalign {\ifnum 0=`}\fi \hrule height 1pt
    \futurelet \reserved@a \@xhline
}
\newcolumntype{"}{@{\hskip\tabcolsep\vrule width 1pt\hskip\tabcolsep}}
\makeatother

\usepackage{pgfplots}
\usepackage{siunitx}
\usepackage{pgfplotstable}

\pgfplotsset{compat=1.13} %
\makeatletter
\pgfplotstableset{
    discard if not/.style 2 args={
        row predicate/.code={
            \def\pgfplotstable@loc@TMPd{\pgfplotstablegetelem{##1}{#1}\of}
            \expandafter\pgfplotstable@loc@TMPd\pgfplotstablename
            \edef\tempa{\pgfplotsretval}
            \edef\tempb{#2}
            \ifx\tempa\tempb
            \else
                \pgfplotstableuserowfalse
            \fi
        }
    }
}
\makeatother

\pgfplotsset{
  discard if/.style 2 args={
    x filter/.code={
      \edef\tempa{\thisrow{#1}}
      \edef\tempb{#2}
      \ifx\tempa\tempb
      \def\pgfmathresult{inf}
      \fi
    }
  },
  discard if not/.style 2 args={
    x filter/.code={
      \edef\tempa{\thisrow{#1}}
      \edef\tempb{#2}
      \ifx\tempa\tempb
      \else
      \def\pgfmathresult{inf}
      \fi
    }
  }
}
\pgfplotstableread[col sep = semicolon]{biconnected-2-club.csv}\dataBiconTwoClubs
\pgfplotstableread[col sep = semicolon]{compare-2-clubs.csv}\dataTwoClubs
\pgfplotstableread[col sep = semicolon]{compare-models.csv}\dataAllModels
\pgfplotstableread[col sep = semicolon]{graphs.csv}\dataGraphs

\newcommand{\tHerTwoClub}{\textsc{$t$-Hereditary $2$-Club}\xspace}
\newcommand{\tConTwoClub}{\textsc{$t$-Connected $2$-Club}\xspace}
\newcommand{\tRobTwoClub}{\textsc{$t$-Robust $2$-Club}\xspace}

\newtheorem{lemma}{Lemma}
\newtheorem{theorem}{Theorem}
\newtheorem{corollary}{Corollary}
\newtheorem{observation}{Observation} 
\newtheorem{definition}{Definition} 
\newtheorem{rrule}{Reduction Rule}

\newcommand{\runtimeDiagram}[1]{
	\begin{tikzpicture}[scale=1]
		\begin{axis}[
					width=0.48\textwidth,
					height=0.4\textwidth,
					xlabel={$t$},
					ylabel={Running time [sec]},
					xmode=log,
					ymode=log,
					legend style = {
									at={(0.5, 1.07)},
									anchor={south},
									font = \small
					},
					legend cell align = left,
					legend columns = 2,
					]
			\addplot[red,mark=triangle*,discard if not={file}{#1}] table [y=timeHereditary, x=t,col sep = semicolon]{compare-models.csv};
			\addplot[blue,mark=star,discard if not={file}{#1}] table [y=timeRobust, x=t,col sep = semicolon]{compare-models.csv};
			\addplot[black,mark=o,discard if not={file}{#1}] table [y=timeConnected, x=t,col sep = semicolon]{compare-models.csv};
		\end{axis}
	\end{tikzpicture}
}

\newcommand{\sizeDiagram}[1]{
	\begin{tikzpicture}[scale=1]
		\begin{axis}[
					width=0.47\textwidth,
					height=0.4\textwidth,
					xlabel={$t$},
					ylabel={2-club size},
					xmode=log,
					legend style = {
									at={(0.5, 1.07)},
									anchor={south},
									font = \small
					},
					legend cell align = left,
					legend columns = 2,
					]
			\addplot[red,mark=triangle*,discard if not={file}{#1}] table [y=sizeHereditary, x=t,col sep = semicolon]{compare-models.csv};
			\addplot[blue,mark=star,discard if not={file}{#1}] table [y=sizeRobust, x=t,col sep = semicolon]{compare-models.csv};
			\addplot[black,mark=o,discard if not={file}{#1}] table [y=sizeConnected, x=t,col sep = semicolon]{compare-models.csv};
		\end{axis}
	\end{tikzpicture}
}

\begin{document}

\title{Exact Algorithms for Finding Well-Connected 2-Clubs in Real-World Graphs: Theory and Experiments\footnote{Parts of this work are based on the last author's master thesis~\cite{Pic15}. Work started when all authors were with TU~Berlin.}}

\author{Christian Komusiewicz$^{1,}$\footnote{CK was partially supported by the DFG, project MAGZ (KO 3669/4-1).} \and André~Nichterlein$^2$ \and Rolf~Niedermeier$^2$ \and Marten~Picker~$^2$}
\date{$^1$ Fachbereich Mathematik und Informatik, Philipps-Universität Marburg, Germany, \\
$^2$ Algorithmics and Computational Complexity, Faculty IV, TU Berlin, Germany,\\
\texttt{\small komusiewicz@informatik.uni-marburg.de \{andre.nichterlein,rolf.niedermeier\}@tu-berlin.de}}

\maketitle

\begin{abstract}
	Finding large ``cliquish'' subgraphs is a central topic in graph mining and community detection. 
	A popular clique relaxation are 2-clubs: instead of asking for subgraphs of diameter one (these are cliques), one asks for subgraphs of diameter at most two (these are 2-clubs). 
	A drawback of the 2-club model is that it produces star-like hub-and-spoke structures as maximum-cardinality solutions. 
	Hence, we study 2-clubs with the additional constraint to be well-connected. 
    More specifically, we investigate the algorithmic complexity for three variants of well-connected 2-clubs, all established in the literature: robust, hereditary, and ``connected'' 2-clubs. 
	Finding these more cohesive 2-clubs is NP-hard; nevertheless, we develop an exact combinatorial algorithm, extensively using efficient data reduction rules. 
	Besides several theoretical insights we provide a number of empirical results based on an  engineered implementation of our exact algorithm. 
	In particular, the algorithm significantly outperforms existing algorithms on almost all (sparse) real-world graphs we considered.
\end{abstract}

\section{Introduction}
\label{sec:intro}
The $s$-club model, introduced by~\citet{mokken}, is a well-es\-tab\-lished mathematical  model for community mining in graphs~\cite{For10}. 
An important special case is the 2-club model. 
In the corresponding algorithmic problem, the task is to compute a largest 2-club, that is, a maximum-cardinality vertex subset inducing a subgraph of diameter  at most two. 
As a community model, 2-clubs have the drawback that they often form hub-and-spoke structures, that is, they often consist of one vertex that is adjacent to all other vertices of the community plus only few additional edges. 
Indeed, previous experimental work revealed that most maximum-cardinality 2-clubs in real-world graphs have exactly this hub-and-spoke structure and are thus fairly sparse~\cite{hartung_experiments}.
This is essentially the opposite of cliques which provide a subgraph model with maximal (edge) density.  
In this sense, standard 2-clubs are far away from being cliquish. 
The drawback of cliques, however, is that they are too demanding with respect to density in order to be meaningful in many applications based on mining cohesive subgraphs~\cite{PYB13,Kom16}.
Hence, we study three 2-club models that exclude the low degree of connectivity that 2-clubs may have while still providing meaningful community abstractions.

\paragraph*{Three 2-Club Models}
In the classic 2-club model, one relaxes the clique demand by allowing that not all vertices of the community are adjacent but that they may also have distance two or, equivalently, that every vertex pair is either adjacent or has a common neighbor. 
The problem with this simple relaxation is that there may be only one vertex which is the common neighbor of all nonadjacent vertices. 
Hence, removing this vertex may destroy connectivity of the 2-club. 
Next, we describe three ``well-connected'' variants of 2-clubs that have been proposed in the literature~\cite{PYB13,VB12}; these three variants will be the central subjects of our studies.

In the first variant of well-connected 2-clubs, one demands that every pair of adjacent vertices has~$t-1$ common neighbors and every pair of nonadjacent vertices has~$t$ common neighbors.

\begin{definition}[\cite{VB12}]\label{def:robust}
  A vertex set~$S \subseteq V$ is a \emph{$t$-robust 2-club} in a graph~$G=(V,E)$ if any pair of vertices in~$S$ is connected in~$G[S]$ by~$t$ internally vertex-disjoint paths of length at most two.
\end{definition} 
In the second variant, one demands that the deletion of few vertices does not destroy the 2-club property.

\begin{definition}[\cite{PYB13}]
	\label{def:t-her} A vertex set~$S \subseteq V$ is a \emph{$t$-hereditary 2-club} in a graph~$G=(V,E)$ if~$G[S\setminus U]$ is a 2-club for all~$U\subset S$ where~$|U|\le t$.
\end{definition}

In the third variant, one only demands that the deletion of few vertices does not destroy connectivity.
\begin{definition}[\cite{PYB13}]\label{def:t-conn} 
	A vertex set~$S \subseteq V$ is a \emph{$t$-connected 2-club} in a graph~$G=(V,E)$ if it is a 2-club and~$G[S]$ is $t$-connected, that is,~$|S|>t$ and~$G[S\setminus U]$ is connected for all~$U\subseteq S$ where~$|U|< t$.
\end{definition}

Observe the difference of one in the sizes of the sets~$U$ in \cref{def:t-her,def:t-conn}. 
Refer to \Cref{fig:model-comparison} for a comparison of the three variants of well-connected 2-clubs.
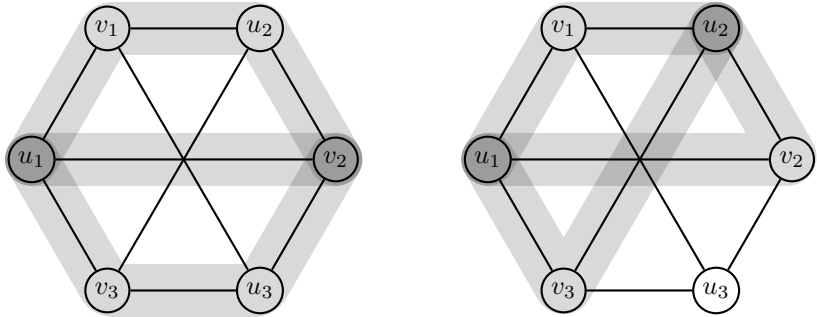
\begin{figure}
	\centering
	\begin{tikzpicture}[>=stealth,rounded corners=2pt,vert/.style={circle,thick,draw,inner sep=2pt,minimum size=2.5mm}]
		\coordinate (a) at (0,0);
		\coordinate (b) at (0,1);
		\foreach[count=\i] \x in {1,1,2,2,3,3}{
			\node(a\i) at ($ (b)!2cm!330-\i*60:(a) $) [vert] {\ifodd\i{$u_\x$}\else{$v_\x$}\fi};
		}
		\foreach \i in {1,3,5}{
			\foreach \j in {2,4,6}{
				\draw [-,thick] (a\i) -- (a\j);
			}
		}
		\tikzstyle{edge} = [color=black,opacity=.15,line cap=round, line join=round, line width=20pt]
		\draw [edge] (a1.center) -- (a6.center) -- (a5.center) -- (a4.center);
		\draw [edge] (a1.center) -- (a2.center) -- (a3.center) -- (a4.center);
		\draw [edge] (a1.center) -- (a4.center);
		
		\begin{scope}[xshift=6cm]
			\coordinate (a) at (0,0);
			\coordinate (b) at (0,1);
			\foreach[count=\i] \x in {1,1,2,2,3,3}{
				\node(a\i) at ($ (b)!2cm!330-\i*60:(a) $) [vert] {\ifodd\i{$u_\x$}\else{$v_\x$}\fi};
			}
			\foreach \i in {1,3,5}{
				\foreach \j in {2,4,6}{
					\draw [-,thick] (a\i) -- (a\j);
				}
			}
			\tikzstyle{edge} = [color=black,opacity=.15,line cap=round, line join=round, line width=20pt]
			\draw [edge] (a1.center) -- (a2.center) -- (a3.center);
			\draw [edge] (a1.center) -- (a4.center) -- (a3.center);
			\draw [edge] (a1.center) -- (a6.center) -- (a3.center);
		\end{scope}
	\end{tikzpicture}
	\caption{
		A complete bipartite graph~$K_{3,3}$ where the three internally vertex-disjoint paths from~$u_1$ to~$v_2$ (left side) and~$u_1$ to~$u_2$ (right side) are highlighted.
		Note that two of the paths on the left side are of length three. 
		Thus, a $K_{3,3}$ is a $1$-robust 2-club but not a $2$-robust 2-club.
		It is also a $2$-hereditary 2-club as deleting two vertices leaves either a~$K_{1,3}$ or a~$K_{2,2}$, both of which are 2-clubs.
		Finally, it is a $3$-connected 2-club (the pictures above cover all cases up to symmetry).
	}
	\label{fig:model-comparison}
\end{figure}
Assuming 2-club sizes~$|S|\geq t+2$, it is easy to see that a $t$-robust 2-club is a $(t-1)$-hereditary 2-club which again is a $t$-connected 2-club.
Clearly, for all three models, fulfilling the connectivity demands for some~$t$ implies that they are also fulfilled for all~$t'\le t$. 
Hence, the size of maximum-cardinality 2-clubs is nonincreasing with increasing~$t$. 

Our three central computational problems, formulated as decision problems (which is more suitable for statements on computational complexity), are defined as follows.

\begin{quote}
  \textsc{$t$-Robust / $t$-Hereditary / $t$-Connected 2-Club}\\
  \textbf{Input:} An undirected graph $G = (V,E)$
  and nonnegative integers $t$ and~$k$. \\
  \textbf{Question:} Does~$G$ contain a $t$-robust / $t$-hereditary / $t$-connected 2-club of size at least~$k$?
\end{quote}

For brevity, when making general statements holding for all three variants, then we refer to them as~$t$-well-connected 2-clubs. 

\paragraph*{Related Work}
\label{sec:related}
Dense subgraph discovery or mining cohesive subgraphs is 
an active field in graph mining
research~\cite{LRJA10,PYB13,Kom16}.  Algorithms, complexity studies, and experiments for
finding $s$-clubs and in particular 2-clubs play a significant role in this
context~\cite{BB17,bourjolly_exact,BS17,carvalho,chang,GHKR14,hartung_experiments,hartung_structural,MB12,schaefer2012,VB12,YPB17}.
\citet{VB12} introduced and motivated the concept of $t$-robustness and~performed experiments
for~$t=2$ and $t=3$, showing that the predicted communities 
for these settings are much larger than
maximum cliques.  
\citet{PYB13} discussed several further variants of $s$-clubs. 

For~$t$-connected 2-clubs, the special case~$t=2$ demands that the $2$-club is biconnected. 
This case was studied by~\citet{YPB17} who presented a branch-and-bound and a branch-and-cut algorithm for the problem of finding biconnected 2-clubs. 
Moreover, \citet{YPB17} showed that the problem of finding a largest biconnected $2$-club is NP-hard.

Finally, \citet{CA17} studied a variant of $s$-clubs where each vertex is demanded to be part of at least one triangle.

There is a large body of previous work that considers the standard 2-club model. 
In particular, several algorithmic approaches (heuristics, exact algorithms) have been proposed and examined for the NP-hard computational task~\cite{bourjolly_exact} to find maximum-cardinality 2-clubs~\cite{bourjolly_exact,bourjolly_heuristics,BS17,chang,hartung_experiments,MB12,schaefer2012}.  

Of particular importance to our work is an exact algorithm  for \textsc{2-Club} that uses the following branching: If the graph contains two vertices~$u$ and~$v$ that have distance at least three, then branch into the two cases to remove~$u$ or~$v$ from~$G$. This algorithm was first proposed by~\citet{bourjolly_exact}. Later, it was shown that this branching gives a fixed-parameter algorithm for the parameter~$|V|-k$ with running time~$O(2^{|V|-k}nm)$~\cite{schaefer2012} and an algorithm with running time~$O(\alpha^n)$ where~$\alpha<1.62$ is the golden ratio~\cite{chang}. A combination of this branching algorithm with data reduction and pruning rules achieves the so far best performance on random networks and on sparse real-world networks~\cite{hartung_experiments}. %

\paragraph*{Our Contributions}
\label{sec:results}
On a conceptual level, we provide an alternative characterization of
$t$-hereditary 2-clubs (see \cref{lem:hereditary-characterization}).
Moreover, we present a ``unifying view'' on all three considered models 
based on ``compatible vertices'' (\cref{def:compatible}).

On the level of algorithm theory, to the best of our knowledge we provide 
the first (formal) NP-hardness proofs for \tRobTwoClub and \tHerTwoClub 
for all~$t\geq 1$. 
The corresponding reduction also yields an exponential 
running time lower bound based 
on the Strong Exponential Time Hypothesis (SETH). 
Moreover, we generalize the NP-hardness result 
for \tConTwoClub due to \citet{YPB17} by showing NP-hardness for 
all $t\geq 1$ even when restricted to split graphs.
On the positive side, we generalize the above-mentioned fixed-parameter algorithm for \textsc{2-Club} parameterized by~$|V|-k$~\cite{bourjolly_exact,chang,hartung_experiments,schaefer2012}, to obtain a fixed-parameter 
algorithm for all three problem variants.

On an algorithm engineering and empirical level, under heavy use of efficient and effective data reduction rules and using the above-indicated ``unified view'' on the three problem variants, we develop an implementation of our fixed-parameter algorithm that outperforms existing implementations on almost all (large and sparse) real-world graphs we experimented with.
Only for random graphs we are clearly beaten by previous implementations.

\section{Preliminaries}
\label{sec:prelim}
We only consider simple undirected graphs~$G=(V,E)$ where $V$ is the vertex set and $E\subseteq \{\{u,v\}\mid u,v\in V\}$ is the edge set. 
Throughout this work, we use $n:= |V|$ and $m:= |E|$ to denote the number of vertices and the number of edges in the input graph~$G$. A \textit{subgraph} $(V', E')$ of a graph $G = (V, E)$ is a graph with $V' \subseteq V$ and $E' \subseteq E$ such that all edges in $E'$ are between vertices in $V'$, i.\,e.\ a graph derived from $G$ by deleting vertices and edges. 
For a graph $G = (V, E)$ and a subset $S \subseteq V$ of vertices, $G[S]:=(S,\{\{u,v\}\in E\mid u,v\in S\})$ denotes the \emph{subgraph induced by $S$}.  
For~$v \in V$ we set~$G-v := G[V\setminus \{v\}]$.
For~$F \subseteq E$ we set $G-F := (V, E \setminus F)$.

A \textit{path} is a sequence of vertices such that no vertex occurs twice and any two successive vertices in the path are adjacent.  
The \textit{length of a path} is the number of edges along it.  
The \textit{distance} $d(u, v)$ between two vertices $u$ and $v$ is the length of a shortest path between these two vertices.  
If there is a path between two vertices, then these vertices are said to be \textit{connected}, and \textit{disconnected} otherwise.  
A graph is \textit{connected} if every pair of its vertices is connected. 
A \textit{connected component} of a graph is a maximal set of vertices which are pairwise connected. 
Unless stated otherwise, we assume that the input graphs are connected, as otherwise we can process each connected component separately. 
The \textit{diameter} of a graph is the length of the longest shortest path in the graph, i.\,e.\ the maximum distance between any two vertices in the graph.  
The \textit{open d-neighborhood} $N_{d}(v)$ of a vertex $v$ is the set of all vertices within distance $d$ of $v$ except $v$ itself.  
The \textit{closed d-neighborhood} $N_{d}[v]$ is defined as $N_{d}[v] := N_{d}(v) \cup \{v\}$. 
We set~$N(v) := N_1(v).$

The \emph{vertex connectivity} of a graph is the minimum number of vertices that has to be removed to disconnect the graph or make it trivial; \emph{edge connectivity} is defined the same way over edges. 
The \emph{edge density} of a graph is $m / {{n}\choose{2}}$, that is, the fraction of present edges compared to the maximum possible number of edges in a graph of $n$ vertices. 
In the \textsc{Clique} problem we are given an undirected graph~$G = (V,E)$ and an integer~$k$, and we ask whether~$G$ contains a clique of size~$k$, that is, a set~$S \subseteq V$ of size~$k$ such that every pair of vertices in~$S$ is adjacent.
Somewhat abusing notation we will use the terms $t$-robust, $t$-hereditary, and $t$-connected 2-club to refer to both the vertex sets and the subgraphs induced by them. 

\section{Structural Properties of Well-Connected 2-Clubs}
\label{sec:properties}

Before studying the computational complexity of finding $t$-well-connected 2-clubs, we first derive some structural properties of the respective $t$-well-connected 2-club models. 
We begin with a simple characterization of $t$-hereditary 2-clubs.
\begin{lemma}\label{lem:hereditary-characterization}
  A vertex set~$S \subseteq V$ is a $t$-hereditary 2-club in a graph~$G=(V,E)$ if and only if each nonadjacent pair of vertices of~$S$ has at least~$t+1$ common neighbors in~$G[S]$.
\end{lemma}
\begin{proof}
  $(\Rightarrow)$ We show the contraposition of this direction. Thus, assume that there
  are two nonadjacent vertices~$u$ and~$v$ in~$S$ that have at most $t$~common
  neighbors. Now let~$U:=N(u)\cap N(v)$ and observe that~$|U|\le t$ by our assumption. The
  set~$S$ is not a 2-club in~$G[V\setminus U]$ because~$u$ and~$v$ are nonadjacent and
  have no common neighbors. Hence,~$S$ is not a $t$-hereditary 2-club.

  $(\Leftarrow)$ Assume that every pair of nonadjacent vertices in~$S$ has at least~$t+1$
  common neighbors. Now, let~$U\subseteq V$ be any vertex set of size at most~$t$. Consider the
  graph~$G[V\setminus U]$. We show that~$S$ is a 2-club in~$G[V\setminus U]$. Let~$u$ and~$v$ be two vertices
  of~$S\setminus U$. If~$u$ and~$v$ are adjacent, then the distance between them is one. Otherwise, the distance between them is two: $u$ and~$v$ have at least $t+1$ common neighbors in~$G$ and one of
  them is not contained in~$U$.
\end{proof}
With this lemma at hand, we can now define the three variants of well-connected 2-clubs based on a predicate over vertex pairs; these definitions are of central importance for specifying our algorithms.
\begin{definition}\label{def:compatible}
	Two vertices $v$~and~$w$ in a graph are called \emph{compatible} 
	\begin{itemize}
		\item for~$t$-robust 2-clubs if they are adjacent and have at least $t-1$ common neighbors, or if they have at least $t$ common neighbors,
		\item for~$t$-hereditary 2-clubs if they are adjacent or if they have at least $t+1$ common neighbors,
		\item for~$t$-connected 2-clubs if they are at distance at most two and are connected by at least $t$~internally vertex-disjoint paths. 
	\end{itemize}
\end{definition}
Note that if there are two compatible vertices $u$ and $v$, then this does not necessarily mean that in a graph~$G$ there is some $t$-well-connected 2-club~$S$ containing both~$u$ and~$v$, see \Cref{fig:compatible} for an example.
\begin{figure}
	\centering
	\begin{tikzpicture}[>=stealth,rounded corners=2pt, vert/.style={circle,thick,draw,inner sep=2pt,minimum size=2.5mm}]

	\foreach \x in {1,2,3}{%
		\node(a\x) at (2*\x,0) [vert] {};
	}
	\foreach[count=\i] \x in {u,v}{%
		\node(b\i) at (2*\i + 1,1) [vert] {$\x$};
		\foreach \y in {1,2,3}{
			\draw [-,thick] (a\y) -- (b\i);
		}
	}
	\end{tikzpicture}
	\caption{
		The two vertices~$u$ and~$v$ are compatible with respect to the $3$-robust 2-club model but the three other vertices are pairwise not
		compatible.
		Thus, there is no~$3$-robust 2-club in the displayed graph.
	}
	\label{fig:compatible}
\end{figure}
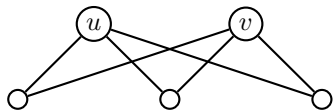
\cref{def:robust,def:t-conn,def:t-her,def:compatible,lem:hereditary-characterization} imply the following.
\begin{observation}\label{obs:comp-is-club}
	A graph is a $t$-robust, $t$-hereditary, or $t$-connected 2-club if and only if every pair of vertices in $G$ is compatible.
\end{observation}

\Cref{def:compatible,obs:comp-is-club} immediately give the following relation between the three concepts. 
Observe that there is an offset of 1 in the correspondence between~$t$-robust or $t$-connected 2-clubs and~$t$-hereditary 2-clubs.
\begin{observation}\label{obs:model-relations}
  A $t$-robust 2-club is a~$(t-1)$-hereditary 2-club and a $t$-connected 2-club. A
  $t$-hereditary 2-club of size at least~$t+2$ is a $(t+1)$-connected 2-club.
\end{observation}

In \cref{fig:model-comparison} in \cref{sec:intro}, we can already see that a complete bipartite~$K_{t,t}$ is a $(t-1)$-hereditary 2-club but not even a $2$-robust 2-club.
Moreover, \cref{fig:t-connected-but-not-2-robust} depicts a $t$-connected 2-club that is neither a $2$-robust nor a~$1$-hereditary 2-club.
Thus, \cref{obs:model-relations} lists all relations between the three $t$-well-connected 2-club models we study.

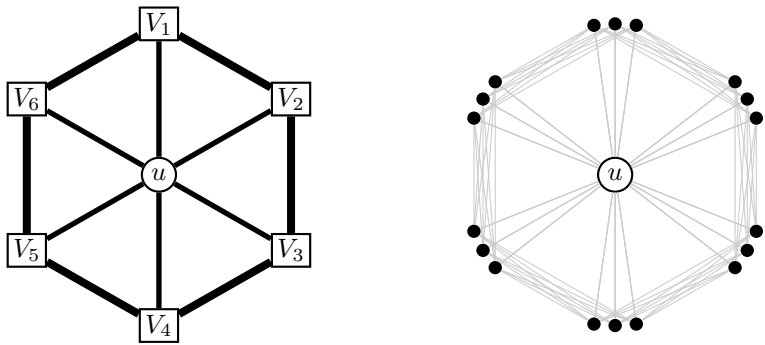
\begin{figure}
	\def\maxR{3}
	\centering
	\begin{tikzpicture}[]
		\coordinate (a) at (0,0);
		\coordinate (b) at (0,1);
		\foreach[count=\nr] \i in {0,...,5}{
			\node(a\i) at ($ (b)!2cm!-\i*60+180:(a) $) [draw,thick,inner sep=2pt,minimum size=5pt] {$V_\nr$};
		}
		\node(c) at (b) [circle,draw,thick,inner sep=2pt,minimum size=5pt] {$u$};
		\foreach \i in {0,...,5}{
			\pgfmathtruncatemacro{\j}{mod(\i + 1,6)};
			\draw [-,line width=3pt] (a\i) -- (a\j);
			\draw [-,line width=2pt] (a\i) -- (c);
		}
		
		\begin{scope}[xshift=6cm]
			\def\smallDeg{4}
			\coordinate (a) at (0,0);
			\coordinate (b) at (0,1);
			\foreach \r in {1,...,\maxR}{
				\foreach \i in {0,...,5}{
					\node(a\i\r) at ($ (b)!2cm!\i*60+\smallDeg*\r*2-\smallDeg-\smallDeg*\maxR:(a) $) [circle,fill=black,inner sep=0pt,minimum size=5pt] {};
				}
			}
			\node(c) at (b) [circle,draw,thick,inner sep=2pt,minimum size=5pt] {$u$};
			\foreach \i in {0,...,5}{
				\pgfmathtruncatemacro{\j}{mod(\i + 1,6)};
				\foreach \r in {1,...,\maxR}{
					\foreach \q in {1,...,\maxR}{
						\draw [-,color=black!20] (a\i\q) -- (a\j\r);
						\draw [-,color=black!20] (a\i\q) -- (c);
					}
				}
			}
		\end{scope}
	\end{tikzpicture}
	\caption{
		Left: A schematic graph where each of~$V_1, V_2, \ldots, V_6$ represents a set of $t$~vertices, for some arbitrary~$t \ge 1$. 
		Each edge~$\{v,w\}$ in the left graph represents a complete bipartite graph with the partite sets represented by~$v$ and~$w$.
		Right: The actual graph for~$t=\maxR$.
		The left graph is a 2-club since~$u$ is adjacent to all vertices.
		Furthermore, the graph is~$t$-connected, that is, each pair of vertices is connected via $t$~internally vertex-disjoint paths.
		Thus the graph is a $t$-connected 2-club.
		Deleting~$u$ results in a graph of diameter three.
		Hence, the graph is not a $1$-hereditary 2-club and not a $2$-robust 2-club.
	}
	\label{fig:t-connected-but-not-2-robust}
\end{figure}

In our implementation, we use data reduction rules that remove vertices of low degree. The
rationale behind these rules is provided by the following observations. 
\begin{observation} \label{obs:smallDegree}
	\begin{enumerate}
		\item Neither a $t$-robust 2-club nor a $t$-connected 2-club can contain a vertex of degree less than $t$.
		\item A $t$-hereditary 2-club $S$ containing a vertex of degree less than $t+1$ is a clique.
	\end{enumerate}
\end{observation}
\begin{proof}
  First, consider $t$-robust and $t$-connected 2-clubs. 
  Let~$S$ be a vertex set containing a vertex~$u$ of degree less than~$t$. 
  The number of vertex-disjoint paths (of any length) between~$u$ and any other vertex in~$S$ is less than~$t$ since~$u$ has less than~$t$ neighbors. 
  Hence,~$S$ is neither a $t$-robust 2-club nor a $t$-connected 2-club.

  We now show the claim for $t$-hereditary 2-clubs. 
  Let~$S$ be a $t$-hereditary 2-club containing a vertex~$u$ of degree less than~$t+1$. 
  It follows from \cref{def:compatible,obs:comp-is-club} that~$u$ is adjacent to all other vertices in~$S$.
  Hence, $|S| \le t+1$ and all vertices in~$S$ have also degree less than~$t+1$.
  It follows from \cref{def:compatible,obs:comp-is-club} that~$S$ is a clique.
\end{proof}
Thus, for~$t$-connected and~$t$-robust 2-clubs one may remove all vertices of degree less than~$t$.  
Furthermore, to find $t$-hereditary 2-clubs containing vertices of degree less than~$t+1$ it is sufficient to solve \textsc{Clique}:

\begin{observation} \label{cliquesGood}
	Every clique is a~$t$-hereditary 2-club for all values of~$t$.
\end{observation}
\begin{proof}
	Clearly, a clique~$K$ is a 2-club. 
	Since~$K$ does not contain any nonadjacent vertices,~\cref{lem:hereditary-characterization} implies that~$K$ is a $t$-hereditary 2-club. 
\end{proof}

\cref{cliquesGood} implies that the largest $t$-hereditary 2-club is at least as large as a maximum-cardinality clique.
In contrast, for every graph there are values of~$t$ such that the graph neither contains a $t$-robust 2-club nor a~$t$-connected 2-club. 
For example a clique of~$n$ vertices is not $n$-robust and also not~$n$-connected.

For~$t=n-1$, where~$n$ is the number of graph vertices, we always obtain the clique model as the following observation implies.
\begin{observation} \label{smallClubs}
	A $t$-hereditary 2-club~$S$ with at most $t+2$ vertices is a clique. 
	A $t$-robust/$t$-connected 2-club~$S$ with at most $t+1$ vertices is a clique of size~$t+1$.
\end{observation}
\begin{proof}
  First consider a $t$-hereditary 2-club~$S$ with at most $t+2$ vertices.  
  For any pair of vertices of~$S$, there are only $t$ further vertices in $S$ which potentially can be common neighbors for both.  
  Therefore, by \cref{def:compatible}, any two vertices in $S$ can only be compatible if they are adjacent. 
  Thus, $S$ is a clique.

  For $t$-robust 2-club and $t$-connected 2-clubs the statement follows from \cref{obs:smallDegree} and the fact that the maximum degree in a graph on~$t+1$ vertices is at most~$t$.
\end{proof}

\section{Complexity Classification and Exact Algorithms}
\label{sec:param-comp}
In this section, we show that all three $t$-well-connected 2-club detection problems are NP-complete for all fixed~$t$. 
In the case of $t$-robust and $t$-hereditary 2-clubs, we observe that the NP-completeness also holds for graphs that share properties with scale-free networks. 
The NP-hardness of the problems motivates the development of fixed-parameter algorithms for them.
Fixed-parameter algorithms are exact algorithms whose running time is polynomial in the
input size but usually exponential in a problem-specific
parameter~\cite{Cyg15,downeyfellows,Nie06,FG06}.  We will describe such an algorithm for
all three problems and show that, in the case of \tRobTwoClub and \tHerTwoClub, it is likely to be optimal with respect to worst-case running time. 
Subsequently, we will show how to solve all three problems by solving~$O(n)$ smaller instances of the
problem. 
This, together with the fixed-parameter algorithm, will serve as the basis for our implementation, providing an exact combinatorial algorithm for finding well-connected 2-clubs.

\subsection{Hardness Results} \label{NPhard} 
By a polynomial-time reduction from $2$-\textsc{Club}, we show, for all values of~$t$, the NP-hardness of finding~$t$-hereditary and $t$-robust 2-clubs.
\begin{theorem} \label{NPc}
	\tRobTwoClub is NP-complete for every~$t\ge 1$.
	\tHerTwoClub is NP-complete for every~$t\ge 0$.
\end{theorem}
\begin{proof}
  Containment in NP is obvious. 
  The NP-hardness of \textsc{$1$-Robust $2$-Club} and \textsc{$0$-Hereditary $2$-Club} is already known as these problems are exactly the 2-\textsc{Club} problem~\cite{bourjolly_exact}.
  The NP-hardness of both problems for the remaining values of~$t$ can be shown via simple polynomial-time reductions from the 2-\textsc{Club} problem, which was shown to be
  NP-complete by \citet{bourjolly_exact}.
  
  For any 2-\textsc{Club} instance~$(G, k)$ we construct an equivalent \tHerTwoClub~instance~$(G', k+t)$, $t > 0$, where~$G' = (V', E')$ is obtained from~$G$ by adding a
  set~$V^*=\{v^{*}_{1}, \ldots , v^{*}_{t}\}$ of $t$ further~vertices and making them adjacent to all vertices of $G$ and to each other. 
  This can obviously be done in polynomial time. 
  It remains to be shown that $G$ has a 2-club of size at least~$k$ if and only if $G'$ has a $t$-hereditary 2-club of size at least~$k+t$. 

  First, consider a 2-club~$S$ in $G$. Then, the set~$S \cup V^{*}$ is a $t$-hereditary 2-club in $G'$: 
  Only vertices in~$S$ can be pairwise nonadjacent.
  Moreover, each pair of nonadjacent vertices in~$S$ has at least one common neighbor in~$S$ and~$t$ common neighbors in~$V^*$. 
  Thus, by \cref{def:compatible,obs:comp-is-club}, $S \cup V^{*}$ is a $t$-hereditary 2-club.

  Conversely, consider a $t$-hereditary 2-club $S'$ of size at least~$k+t$ in~$G'$. 
  Note that~$|V^*| = t$ and thus, by \Cref{def:t-her}, $S' \setminus V^*$ is a 2-club of size at least~$k$ in~$G$.

  For~\tRobTwoClub, $t > 1$, we adapt the construction by adding to~$G$ a
  set~$V^*=\{v^{*}_{1}, \ldots , v^{*}_{t-1}\}$ of $t-1$ further~vertices and making them
  adjacent to all vertices of~$V\cup V^*$.  

If~$G$ has a 2-club~$S$ of size at least~$k$,
  then~$S\cup V^*$ is a $t$-robust 2-club of size at least~$k+{t-1}$ in~$G'$: Every pair of
  vertices from~$S$ is connected by at least~$t$ internally vertex-disjoint paths of
  length two in~$G'[S\cup V^*]$ since they have~$t-1$ such paths via the vertices in~$V^*$
  and are either adjacent or have a common neighbor in~$S$. 

Conversely, let~$S'$ be a
  $t$-robust 2-club of size at least~$k+t-1$ in~$G'$ and let~$S:=S'\cap V$. Then,~$S$ is a 2-club
  of size at least~$k$ in~$G$: 
  Every pair of nonadjacent vertices in~$S$ is connected by $t$ vertex-disjoint paths of length two in~$G'[S']$.
  At least one of these paths contains only vertices from~$S$ since~$|S'\setminus S|\le t-1$. Hence, these nonadjacent vertices
  have a common neighbor in~$S$.
  
   Altogether we have polynomial-time reductions from 2-\textsc{Club} to \tHerTwoClub for every~$t\ge 0$ and from 2-\textsc{Club} to \tRobTwoClub for every~$t\ge 1$.
\end{proof}

The above reduction leaves the graph almost unchanged and thus many NP-hardness results that were obtained for 2-\textsc{Club} transfer almost directly to \tHerTwoClub and~\tRobTwoClub. 
We list here the hardness results for graphs with properties that are typical for certain real-world networks. %
For instance, the class of split graphs finds applications in social networks in the context of ``core-periphery models''~\cite{BE00}.

The graph classes mentioned in the subsequent corollary are defined as follows.  
A graph~$G=(V,E)$ is a \emph{split graph} if~$V$ can be partitioned into~$V_1$ and~$V_2$ such that~$G[V_1]$ is a clique and~$G[V_2]$ is an independent set.  
A graph has \emph{domination number~$1$} if there is a vertex that is adjacent to all other vertices.  
A graph~$G$ has \emph{degeneracy~$d$} if every subgraph of~$G$ contains a vertex of degree at most~$d$.

\begin{corollary} \label{structuralParam}
	\tHerTwoClub and \tRobTwoClub are NP-complete for~$t \ge 1$ on 
	\begin{enumerate}[label=\textnormal{(\arabic*)}]
		\item graphs with both diameter two and domination number one,\label{item:diam}
		\item connected graphs with average vertex degree at most $\alpha$, for all~$\alpha > 2$,\label{item:avg-deg}
		\item graphs with degeneracy $6 + t$, and \label{item:degen}
		\item split graphs. \label{item:split-graphs}
	\end{enumerate}
\end{corollary}
\begin{proof}
Regarding~\ref{item:diam}, note that the graph $G'$ constructed by the reduction in the proof of \cref{NPc} has at least one vertex adjacent to all other vertices and thus diameter two and domination number one. %

Regarding~\ref{item:avg-deg}, note that the statement is known for $1$-robust 2-clubs, that is, standard 2-clubs~\cite{hartung_structural}.
Thus, we discuss subsequently $t$-hereditary 2-clubs for~$t\ge 1$ and $t$-robust 2-clubs for~$t \ge 2$. 
We add a long path to the graph obtained by the construction in the proof of \cref{NPc} and add an edge from one endpoint of the path to an arbitrary vertex in the original graph.
This decreases the average degree to any constant arbitrarily close to~2. 
No vertex of the path can be contained in any $t$-hereditary ($t\ge 1$) or $t$-robust ($t \ge 2$) 2-club of size at least three.
Since we can assume that the solution size~$k$ is at least three, we obtain an equivalent instance with the desired average degree. 
This adjusted reduction shows NP-hardness for any average degree greater than two.

Regarding \ref{item:degen}, note that 2-\textsc{Club} is NP-hard on graphs with degeneracy six~\cite{hartung_structural}. 
Adding $t$ vertices that are adjacent to all other vertices increases the degeneracy by at most $t$,
therefore \tHerTwoClub and \tRobTwoClub are NP-hard on graphs with degeneracy $6 + t$. %

Regarding \ref{item:split-graphs}, observe that if the graph~$G$ in the 2-\textsc{Club} instance is a split graph, then also the graph~$G'$ constructed by the reduction in the proof of \cref{NPc} is a split graph.
This implies the NP-hardness of \tHerTwoClub and~\tRobTwoClub on split graphs since 2-\textsc{Club} remains NP-hard on split graphs~\cite{asahiro}.
\end{proof}
\tConTwoClub is known to be NP-hard for~$t\le 2$~\cite{BBT05,YPB17}; for the sake of
completeness, we show hardness for all constant~$t$. Observe in this context that the reduction
in the proof of Theorem~\ref{NPc} does not work for \tConTwoClub: Adding $t-1$ common neighbors
for all vertices of~$G$ results in a~$t$-connected graph of diameter two.
\begin{theorem}
  \tConTwoClub is NP-complete, even on split graphs, for all~$t\ge 1$.
\end{theorem}
\begin{proof}
  Containment in NP is obvious. To show NP-hardness, we use a known reduction from
  \textsc{Clique} to \textsc{2-Club}~\cite{BBT05}. We assume, however, that in the
  \textsc{Clique} instance every vertex has at least~$t$ neighbors; obviously, \textsc{Clique} remains
  NP-hard with this restriction. Given such an instance~$(G=(V,E),k)$ of \textsc{Clique},
  construct an instance~$G'=(V',E')$ of \tConTwoClub as follows. The vertex set~$V'$ consists
  of~$V$ and one vertex for each edge of~$G$. More formally,~$V':=V\cup V_E$ where~$V_E:=\{v_e\mid e\in E\}$. Now construct the
  edge set~$E'$: First, make~$V_E$ a clique. Second, for each~$v\in V$
  and each edge~$e$ that is incident with~$v$ in~$G$, add an edge between~$v$ and~$v_e$. The
  graph~$G'$ is a split graph and can obviously be constructed in polynomial time. We conclude
  the proof by showing that~$G$ has a clique of size at least~$k$ if and only if~$G'$ has a
  $t$-connected 2-club of size at least~$k+|E|$.

  Let~$S$ be a clique in~$G$. Then~$S\cup V_E$ is a $t$-connected 2-club of size at
  least~$k+|E|$ in~$G'$: 
  $t$-connectivity follows from the fact that, by the assumption on the minimum degree in~$G$, each vertex of~$S$ has $t$ neighbors in the large clique~$V_E$. 
  The 2-club property follows from the fact that each pair of vertices~$u, v \in S$ has a common neighbor in~$V_E$ since~$u$ and~$v$ are adjacent in~$G$.

  Conversely, consider a $t$-connected 2-club~$S'$ of size at least~$k+|E|$
  in~$G'$. Consider~$S:=S'\cap V$. Clearly,~$|S|\ge k$. Moreover,~$S$ is a clique in~$G$: every
  pair of vertices~$u$ and~$v$ in~$S$ has a common neighbor in~$G'$. This common neighbor
  represents an edge of~$G$ that is incident with~$u$ and~$v$.
\end{proof}
\subsection{Fixed-Parameter Algorithm with Respect to the Dual Parameter n-k}\label{dualParam}
We now describe a fixed-parameter algorithm for all three problems which is the basis for our implementation. 
The general idea---branching on vertices that are incompatible---was introduced by~\citet{bourjolly_exact} for \textsc{2-Club}.
The analysis of the algorithm exploits the parameter~$\ell:=n-k$, that is, the number of
vertices not contained in a largest $t$-well-connected 2-club.

\begin{theorem} \label{2tC-VD} 
	For~$\ell := n-k$, \tHerTwoClub and \tRobTwoClub can be solved in $O(2^{\ell}\cdot nm)$~time and \tConTwoClub can be solved in~$O(2^{\ell}\cdot t^2nm)$ time. 
\end{theorem}
\begin{proof}
\begin{algorithm}[t!]\small
	\caption{A fixed-parameter algorithm computing a maximum-size $t$-well-connected 2-club parameterized by $\ell:=n-k$.}
	\label{alg:dual-param}
	\KwIn{A graph~$G = (V,E)$ and positive integers~$t$ and~$\ell$.}
	\KwOut{True if~$G$ contains $t$-well-connected 2-club of size at least~$n-\ell$, false otherwise.}
	\SetKwFunction{branch}{branch}
	\SetKwProg{func}{Function}{}{}
	\branch{$G$, $\ell$}\;
	\func{\branch{$G$, $\ell$}}{
		\lIf(){$G$ is a $t$-well-connected 2-club}{\KwRet{true}}
		\lIf(){$\ell = 0$}{\KwRet{false}}
		$u,v \gets $ two incompatible vertices\;
		\KwRet{\branch{$G - u$, $\ell-1$}${}\vee{}$\branch{$G - v$, $\ell-1$}}\;
	}
\end{algorithm}
	The algorithm is a simple branching algorithm that, as long as two vertices are incompatible, branches into deleting either of them; see \Cref{alg:dual-param} for a pseudocode.
	This procedure leads to a search tree with at most $2^{\ell}$ leaves. To determine
        in Line~3 whether the current graph already fulfills the desired property, we check
        whether there are any incompatible vertices in the graph. Thus, the only
        difference between the three algorithms is the definition of compatibility, that
        is, in Lines~3 and~5 we need different algorithms to find incompatible vertices.

        For~\tRobTwoClub and for~\tHerTwoClub, we first count for each vertex pair the
        number of common neighbors. This can be done in~$O(nm)$ time by considering for
        each vertex in~$G$ each pair of neighbors. An incompatible pair can now be found in~$O(n^2)$ time
         by considering each vertex pair and applying
        Definition~\ref{def:compatible}.
        
        For~\tConTwoClub, we first determine in~$O(nm)$ time whether there is a vertex
        pair that has distance at least three in~$G$. If yes, then the two vertices are
        nonadjacent. 
        Otherwise, we apply the algorithm of~Galil~\cite{Gal80} that determines in~$O(t^2nm)$ time whether the graph is~$t$-connected or to find a pair of vertices~$u$ and~$v$ that are in different $t$-connected components if this is not the case.
\end{proof}
We next prove that \Cref{alg:dual-param} is likely to be asymptotically optimal for~\tRobTwoClub and~\tHerTwoClub by showing that there can be no algorithm with running time $(2 - \epsilon)^{\ell}\cdot n^{O(1)}$ for any $\epsilon > 0$ for $(2,t)$-\textsc{Club}, unless the so-called Strong Exponential Time Hypothesis (SETH)~\cite{IPZ01} fails.\footnote{
The SETH is a well-accepted hypothesis in computational complexity theory and is used in several other lower-bound results~\cite{LMS11}.} 
The SETH postulates that \textsc{CNF-Sat}, the satisfiability problem for boolean formulas in conjunctive normal
form, cannot be solved in time $(2 - \epsilon)^{N}\cdot |F|^{O(1)}$ for any
$\epsilon > 0$, where $N$ is the number of variables in the given formula~$F$ and~$|F|$ is the formula size.
The proof is based on the reduction described in the proof of \Cref{NPc} and a similar known result for \textsc{2-Club}~\cite{hartung_experiments}.

\begin{theorem} \label{2tC-VD-SETH}
 If the Strong Exponential Time Hypothesis (SETH) holds, then \tRobTwoClub and \tHerTwoClub do not admit a~$(2 - \epsilon)^{\ell}\cdot n^{O(1)}$-time algorithm for any $\epsilon > 0$.
\end{theorem}

\begin{proof}
  The reduction presented in the proof of~\cref{NPc} transforms a
  \textsc{2-Club} instance~$(G=(V,E),k)$ into a \tRobTwoClub or a \tHerTwoClub instance~$(G'=(V',E'),k')$ where~$|V|'=|V|+t-1$ and $k'=k+t-1$ or~$|V|'=|V|+t$ and $k'=k+t$. 
  Hence, the dual parameters~$\ell=|V|-k$ and~$\ell'=|V'|-k'$ are not changed by either reduction. 
  Consequently, any algorithm solving~\tRobTwoClub or \tHerTwoClub in~$(2 - \epsilon)^{\ell}\cdot n^{O(1)}$ time for some~$\epsilon>0$ implies an algorithm solving~\textsc{$2$-Club} in this time. 
  This is impossible if the SETH is true~\cite{hartung_experiments}.
\end{proof}
As a final note, as shown by \citet{chang} for \textsc{2-Club}, \cref{alg:dual-param} has a worst-case
running time of $\alpha^{n} n^{O(1)}$, where $\alpha < 1.62$ is the golden ratio.
\begin{corollary}\label{goldenratio}
	\tRobTwoClub, \tHerTwoClub, and \tConTwoClub can be solved in $O(1.62^{n})$ time.
\end{corollary}
\begin{proof}
  Consider the search tree algorithm of \Cref{2tC-VD}. Whenever we branch over a pair of
  incompatible vertices $v$ and $w$, we can delete~$v$ in one branch and fix~$v$ (as
  contained in the sought solution) in the other branch. The justification for fixing~$v$
  is that the first branch already explores all possible solutions that do not
  contain~$v$, therefore~$v$ can be assumed to be in the solution in the second
  branch. For any fixed vertex, we delete all vertices that are incompatible with~$v$. Thus,
  in the branch fixing~$v$, we can delete~$w$ and, possibly, some further vertices.
  Branching continues until either (a) the graph is $t$-well-connected 2-club or (b) all remaining vertices are fixed. 
  In case (a) we found a solution and in case (b) we fixed two incompatible vertices, leading to a contradiction. 
  Letting~$j$ denote the number of nonfixed vertices in~$G$, the recursion for the number of leaves in the search tree is~$T(j)=T(j-1)+T(j-2)$ and initially~$j=n$. 
  This is the recurrence for Fibonacci numbers implying a search tree size of~$O(\alpha^n)$ where~$\alpha<1.62$.
\end{proof}

\subsection{Turing Kernelization}\label{maxdegree}
In this section, we describe parameterized preprocessing rules for our three $t$-well-connected 2-club detection problems. 
Herein, the parameter is the maximum
degree~$\Delta$ of the input graph~$G$. 
The idea of a ``Turing kernelization'' is to solve
the original problem instance by solving polynomially many problem
instances whose size is upper-bounded in the parameter value (cf.\ \citet{Kra14}). These problem
instances are small if the parameter value is small. In our case, we
need to solve~$n$ problem instances, each consisting of~$O(\Delta^2)$ vertices.
An analogous result for 2-club was shown by~\citet{schaefer2012}.
\begin{theorem}\label{maxdegthm}
  \tRobTwoClub, \tHerTwoClub, \tConTwoClub can be solved by solving~$O(n)$ instances
  of~\tRobTwoClub, \tHerTwoClub, and \tConTwoClub, respectively, where each instance
  contains at most $\Delta^{2}$~vertices.
\end{theorem}
\begin{proof}
  By definition, every 2-club has diameter at most two and is therefore fully contained in
  the closed 2-neighborhood of each of its vertices in the original input graph. To find
  the largest $t$-well-connected 2-club in a graph~$G$ it is now
  sufficient to search in each of the $n$ closed 2-neighborhoods of~$G$ for a largest such
  2-club and take the largest among all of them. In a graph with maximum degree~$\Delta\ge 2$
  the size of the closed 2-neighborhood of any vertex~$v$ is upper-bounded by~$\Delta^{2}$
  as $v$ can have at most $\Delta$~neighbors which each can have at most~$\Delta - 1$
  other neighbors besides $v$.
\end{proof}
While \cref{maxdegthm} is simple to show, it has important practical
implications. It allows us to decompose larger (and relatively sparse)
graphs into a linear number of smaller graphs where the problem can be
solved efficiently. The smaller individual graphs are called
\textit{Turing kernels}.
This result has been used for 2-\textsc{Club}-implementations \cite{hartung_experiments} and we use it as well in our implementation as described in \Cref{sec:Algorithm}.
As a corollary, we obtain the following running time bound employing the parameter~$\Delta$.
To this end, note that one can compute all Turing kernels in $O(n \Delta^2)$~time.
\begin{corollary}
  \tRobTwoClub, \tHerTwoClub, and \tConTwoClub can be solved in
  $O(1.62^{\Delta^2}\cdot n)$ time on graphs with maximum degree~$\Delta$.
\end{corollary}

\section{Implemented Algorithm} 
\label{sec:Algorithm}

In this section we describe our implementation for solving all three extensions of the \textsc{2-Club} model.
Recall that we use the term $t$-well-connected 2-clubs to make statements that hold for all three models.
The pseudocode given in \Cref{alg:well-connected-two-club} shows the general setup of the algorithm which includes the Turing kernelization (\cref{maxdegthm}) and the search tree (\cref{2tC-VD}) combined with data reduction rules inspired by \citet{hartung_experiments}.
\begin{algorithm}[t!]\small
  \caption{Our algorithm for finding $t$-well-connected 2-clubs.}
 \label{alg:well-connected-two-club}
	\KwIn{A graph~$G = (V,E)$ and a positive integer~$t$.}
	\KwOut{A maximum-size $t$-well-connected 2-club.}
	$\ell \gets 0$; $S \gets \emptyset$ \tcp*{$\ell$: lower bound, $S$: best solution so far}
	\lIf(\tcp*[f]{initialize lower bound}){$t = 1$}{$\ell \gets \Delta + 1$}

	$\mathcal{T} \gets $ Turing kernels sorted by size in nondecreasing order\; \label{line:turingkernel} 
	\ForEach(\label{line:Turing-kernel-loop}){\upshape Turing kernel $T \in \mathcal{T}$}
	{
		Apply data reduction rules on~$T$ \tcp*{see \Cref{sec:reduction-rules}}
		\If{\upshape size of $T$ is larger than $\ell$}{
			solution${}\gets{}$branch on~$T$ \tcp*{using \cref{alg:dual-param}; see \Cref{sec:search-tree}}
			\If{\upshape solution size${}>\ell$}{
				$\ell \gets{}$solution size; $S \gets{}$solution \label{line:found-new-solution}\;
			}
		}
	}
	\KwRet{$S$}
\end{algorithm}
In addition, our algorithm uses data structures for maintaining compatibility information (see \cref{def:compatible}) for all vertices. %

While the criterion for compatibility depends on the selected model, our algorithm mostly works with the compatibility information.
Thus, the major difference between the three models within our algorithm is the test whether two vertices are compatible.
For the $t$-robust 2-club and the $t$-hereditary 2-club model, the compatibility of two vertices depends on whether they are adjacent and how many common neighbors they have. 
It is thus relatively easy to maintain the compatibility information for these two models as deleting a vertex can only affect the compatibilities of its neighbors.
For the $t$-connected 2-club model one needs to compute the number of (internally) vertex-disjoint paths to determine whether two vertices at distance at most two are compatible.
To this end, we use a standard linear-time reduction to a maximum flow problem with unit capacities. 
To compute the maximum flow, we use the classic algorithm of Ford and Fulkerson~\cite{FF56}.

\subsection{The Search Tree Method and Turing Kernelization}\label{sec:search-tree}

The first step of the algorithm is to utilize Turing kernelization based on our observations from \Cref{maxdegree}.  
For each vertex~$v$ of the original graph, we construct a Turing kernel consisting of the closed 2-neighborhood~$N_{2}[v]$ of that vertex.  
We then run the search tree method described in \Cref{dualParam} on each of the generated Turing kernels.  
For each Turing kernel we initially mark the vertex~$v$ from which we derived the kernel (marking means that we assume that~$v$ is part of the $t$-well-connected 2-club). 
The marking is correct since every vertex is marked exactly once.
Additionally, after solving the Turing kernel for~$v$ we remove~$v$ from the graph and from all subsequent Turing kernels because if~$v$ is part of an optimal solution, then we have found this solution in the kernel for~$v$.
Also, we discard any Turing kernel not larger (in terms of number of vertices) than the current lower bound. 
The largest $t$-well-connected 2-club of the input graph then is simply the largest among the $t$-well-connected 2-clubs of each individual kernel.

As an initial lower bound, we use maximum degree plus one for $t = 1$ ($t=0$ in the case of $t$-hereditary 2-clubs). 
For $t \geq 2$, we set the initial lower bound to zero. 
This is because we know of no lower bounds which are easy to compute and sufficiently strong.
Due to the absence of an initial lower bound, we check the individual kernels of the graph in order of nondecreasing size.
The idea is to keep for all considered kernels the gap between the current solution size lower bound and the size of the currently considered kernel as small as possible, as this gap constitutes an upper bound on the depth of the search tree for any given kernel. 
To reduce overhead from re-evaluating the size of each kernel after every vertex deletion, we sort the kernels of the graph once at the beginning of the algorithm by the size of the 2-neighborhood of the respective vertices and stick to this order throughout the run of the algorithm, even though some kernels might become smaller than some of their predecessors due to vertex deletions.
In our experiments we observed that this worked quite well in the sense that on instances where the algorithm needs more than a minute, less than 10\% of the running time of our algorithm is spent on the branching.

\subsection{Data Reduction Rules}\label{sec:reduction-rules}

To further improve the search tree method, we exhaustively apply various data reduction rules to shrink the original graph as well as every Turing kernel. %
The data reduction rules we use are mostly generalizations of known data reduction rules for 2-\textsc{Club}~\cite{hartung_experiments}.
They are as follows:
\begin{rrule}[Marked Incompatible Rule]
	If at any time two vertices are both marked and incompatible, then abort the branch.
\end{rrule}

\begin{rrule}[Incompatible Resolution]
	Remove all vertices which are incompatible with any marked vertices.
\end{rrule}
The correctness of these two rules is obvious as a $t$-well-connected 2-club by definition cannot contain incompatible vertices.

\begin{rrule}[Low Degree Rule] 
	Remove vertices of degree less than~$t$ (less than $t+1$ for the $t$-hereditary $2$-club model).
	For $t = 1$ ($t=0$ for the $t$-hereditary $2$-club model) delete all vertices of degree one. 
	If a marked vertex was deleted, then abort the branch.
\end{rrule}
This data reduction rule is not universally correct, so we need to make two exceptions.
 
First, deleting degree-one vertices for $t = 1$ ($t=0$ for the $t$-hereditary $2$-club model) is only correct when initializing the algorithm with the closed 1-neighborhood of the maximum degree vertex as initial solution. 
A 2-club containing a degree-one vertex can only consist of the closed 1-neighborhood of its neighbor. 
Hence, no solution better than the one obtained from the closed 1-neighborhood of the maximum degree vertex contains degree-one vertices.

Second, \cref{obs:smallDegree} shows that for the $t$-hereditary 2-club model, deleting vertices of degree less than~$t+1$ is only correct if the solution is not a clique.
However, cliques of size at most $t$ are by definition $t$-hereditary 2-clubs. 
Thus, we first assume that the $t$-hereditary 2-club has size at least~$t+1$ which allows us to apply the Low Degree Rule. 
If we do not find a $t$-hereditary 2-club under this assumption, then we use a standard \textsc{Clique} algorithm to find the largest clique and return it. 
The running time of the \textsc{Clique} algorithm is, if it is invoked at all, much lower than the running time for finding the $t$-hereditary 2-club of size at least~$t+1$. 
We use the Low Degree Rule extensively on the whole graph as well as on every kernel individually, both initially and whenever the degree of a vertex falls below~$t$ due to vertex deletions. 

\begin{rrule}[Low Compatibility Rule]
	Remove vertices whose number of compatible vertices is lower than the current lower bound.
	If a marked vertex was deleted, then abort the branch.
\end{rrule}

The correctness of this data reduction rule follows from the fact that a $t$-well-connected 2-club containing some vertex $v$ can only contain vertices which are compatible with~$v$.
Thus, if~$v$ has less compatibilities than the current lower bound, then no $t$-well-connected 2-club containing $v$ contains more vertices than the current lower bound.

Our next data reduction rule uses the notion of vertex cover.
A \emph{vertex cover} in a graph~$G = (V,E)$ is a vertex subset~$V' \subseteq V$ such that each edge in~$E$ is incident with at least one vertex in~$V'$.

\begin{rrule}[Vertex Cover Rule]
  Let $G = (V, E)$ be the current instance and let $G_{C} = (V, E')$ be the corresponding incompatibility graph, where $\{v, w\} \in E'$ if and only if $v$ and $w$ are incompatible. 
  Compute a lower bound~$b$ on the vertex cover size of~$G_C$. 
  If~$|V|-b$ does not exceed our current lower bound for the solution size, then abort the branch.
\end{rrule}
The correctness of the Vertex Cover Rule follows from the observation that the size of a minimum vertex cover of $G_{C}$ is a lower bound on the number of vertices in $G$ that still must be deleted in order to obtain a $t$-well-connected 2-club.
As a quickly-to-compute lower bound for the vertex cover size we use the size of a maximal matching.

The Vertex Cover Rule (VCR) reduces the search space substantially, but it is costly in terms of running time. 
We therefore use the following observations to cut down the number of necessary applications of the VCR: %
Whenever we have created no new incompatibilities since the last application of the VCR, the last result of the VCR is still a valid lower bound for the required number of vertex deletions. 
Thus, we reuse the last lower bound~$b$ on the vertex cover when applying the VCR rule.
To further decrease the number of VCR applications, we remember the last vertex cover of the incompatibility graph and only count as new those incompatibilities which are not covered by this last vertex cover.
If we have created $x$~new conflicts since the last application of the VCR, then a new application of the VCR will report at most $x$~further required deletions over the last result of the VCR. 
So if the last upper bound minus~$x$ is not enough to trigger the VCR, then there is no need to check the VCR again. 

\begin{rrule}[No Choice Rule] 
	If at any time two nonadjacent marked vertices have exactly~$x$ common neighbors, then mark all of their neighbors which are not marked yet. 
	Here, $x$ is:
	\begin{itemize}
		\item $t$ for the $t$-robust 2-club model,
		\item $t+1$ for the $t$-hereditary 2-club model, and
		\item $1$ for the $t$-connected 2-club model.
	\end{itemize}
\end{rrule}
The reason for having a much weaker bound for $t$-connected 2-clubs is that the $t$~paths connecting two compatible vertices can have length more than two:
For example, consider a cycle~$C_5$ on five vertices with exactly two nonadjacent vertices~$u$ and~$v$ being marked. 
In the cycle, $u$ and~$v$ have exactly one common neighbor. 
However, the~$C_5$ is a $2$-connected $2$-club.
The graph shown in \cref{fig:t-connected-but-not-2-robust} in \cref{sec:properties} ``generalizes'' the~$C_5$ example: 
The graph is a $t$-connected $2$-club is displayed where vertex pairs (e.\,g.\ a marked vertex in the set~$V_1$ and a marked vertex in the set~$V_4$) have only one common neighbor.
Thus, if two nonadjacent marked vertices have at least two common neighbors, then one might still be able to delete one of the two vertices without destroying $t$-connectedness.

We exhaustively apply all reduction rules except the Vertex Cover Rule on each constructed Turing kernel.
We use the Vertex Cover Rule only during the branching to prune the search tree.

\subsection{Data Structures} \label{ssec:data-structures}

To perform the data reduction rules quickly at all times, the following information about the current Turing kernel is held up-to-date.
To this end, recall that the current Turing kernel is a vertex~$v \in V$ together with its 2-neighborhood and that the Turing kernels are processed one after the other, see \Crefrange{line:Turing-kernel-loop}{line:found-new-solution} in \Cref{alg:well-connected-two-club}.

\paragraph{Common Neighbor Matrix} 
A matrix storing for each pair of vertices in the current Turing kernel their number of common neighbors. 
This data structure is very important for the $t$-hereditary 2-club and $t$-robust 2-club model. 
However, also in the $t$-connected 2-club model it is useful to check in constant time whether two nonadjacent vertices have a common neighbor, that is, whether they are at distance at most two.

\paragraph{Incompatibility Graph} 
A graph over the same vertex set as in the Turing kernel where two vertices are adjacent if and only if the two vertices are incompatible.
Hence, for each edge $\{v, w\}$ in this graph at least one of the two vertices~$v$ or~$w$ must be deleted in order to obtain a $t$-well-connected 2-club. 

\paragraph{Compatibility Vector} 
A vector containing for each vertex $v$ of the current Turing kernel the number of vertices that are compatible with~$v$. 
Among other things this vector allows a fast (linear in the number of vertices of the current Turing kernel) look-up of the vertex that is compatible with the fewest other vertices; this vertex is used as the next branching candidate.
When looking for 2-clubs, the entries in the compatibility vector correspond to the size of the 2-neighborhood of the respective vertex.

\medskip

The initialization of these data structures requires $O(m\Delta + n^{2})$ time for the $t$-hereditary 2-club model and $t$-robust 2-club model, where $n$ is the number of vertices in the kernel, $m$ is the number of edges in the kernel, and $\Delta$ is the maximum vertex degree in the kernel. 
For the $t$-connected 2-club model we implemented the classical maximum flow algorithm of Ford and Fulkerson~\cite{FF56} to determine whether two vertices are compatible. 
Thus, for $t$-connected 2-clubs there is an additional~$O(tm)$ running time factor giving an overall initialization time of~$O(tn^{2}m)$.

We first initialize the common neighbor matrix with zero values and then for each vertex $v$ of the kernel increase the entry of each pair of neighbors of $v$ by one.
Afterwards, for the $t$-hereditary 2-club model and the $t$-robust 2-club model, a single pass over the common neighbor matrix suffices to count the number of compatible vertices for each vertex and identifying every initial incompatibility. 
For the $t$-connected 2-club model, we run the maximum flow algorithm for every pair of vertices that are at distance at most two from each other.
The initialization cost can be quite large for dense kernels with many vertices (more than 50\% of the overall running time on large social networks), so we apply all data reduction rules we can use without this information first.
If after this quick preprocessing the kernel has less vertices than the current lower bound, then we discard the kernel before the initialization of the data structures.
To keep the information continuously up-to-date, for every vertex $v$ that is deleted, we do the following.

\begin{enumerate}
	\item Decrease the number of compatible vertices of all vertices not compatible with $v$ by one; this can be done in $O(n)$ time, where~$n$ is the number of vertices in the current kernel.
	\item Decrease the common neighbor counter for all pairs of neighbors of~$v$ by one; this can be done in $O(\deg(v)^{2})$ time, where~$\deg(v)$ is the degree of~$v$.
	\item Update the incompatibility graph:
	\begin{itemize}
		\item For $t$-robust 2-clubs and $t$-hereditary 2-clubs, we can quickly update the incompatibility graph while updating the common neighbor counter (see \cref{def:compatible}).
		\item For $t$-connected 2-clubs, we run the maximum flow algorithm between every pair of vertices that were compatible before deleting~$v$.
		If the two vertices have distance more than two, then we set them incompatible without running the maximum flow algorithm.
	\end{itemize}
\end{enumerate}

For $t$-robust 2-clubs and $t$-hereditary 2-clubs, this leads to an~$O(n + \deg(v)^{2})$-time overhead for every deletion of a vertex~$v$ to keep the information up-to-date.
For $t$-connected 2-clubs, the overhead is in the worst case~$O(tn^2m)$.
Overall, the overhead is quite significant but enables us to perform our data reduction rules quickly and extensively whenever they apply and as soon as they apply.

\subsection{Integer Linear Programming Formulation for the \tHerTwoClub Problem}\label{ssec:ilp}

We want to utilize a state-of-the-art ILP solver (Gurobi) as benchmark for our algorithm on instances with $t \geq 2$.
To this end, we now describe the ILP formulation for \tHerTwoClub which we will use in our performance tests in \Cref{sec:experiments}.  
Given a graph~$G = (V,E)$ and an integer~$t$, the \tHerTwoClub problem can be solved with the following ILP:

Introduce a 0/1-variable~$x_{v}$ for each vertex~$v$ of the graph~$G$. 
Setting~$x_{v} = 1$ indicates~$v\in S$, where~$S$ is the solution set of vertices, that is, a set of vertices constituting a maximum $t$-hereditary $2$-club.
The ILP is 
\begin{align*}
	\text{max } {} & {} \sum_{v \in V} x_v \\
	\text{s.\,t. } {} & {} (t+1) \cdot x_{u} + (t+1) \cdot x_{w} - \smashoperator{\sum_{v\in N(u) \cap N(w)}} x_{v} \leq (t+1) && \forall \{u,w\} \notin E, \\
	& x_v \in \{0,1\} && \forall v \in V.
\end{align*}
For each nonadjacent pair of vertices~$u$ and~$w$ we have a constraint 
enforcing to either take at most one of the two vertices~$u$ and~$w$ into the solution or add at least $t+1$ of their common neighbors, too.

This formulation generalizes the ILP formulation for the 2-\textsc{Club} problem as first described by~\citet{bourjolly_exact} and used several times afterwards~\cite{BBT05,carvalho,BS17}. 
Furthermore, we remark that the ILP formulation of \citet{VB12} for finding $t$-robust $2$-clubs is very similar.

\section{Computational Experiments}
\label{sec:experiments}

In this section, we provide an experimental evaluation of our implementation\footnote{Available from~\url{http://fpt.akt.tu-berlin.de/software/well-connected-2-club/wc2club.jar}.} (see \cref{sec:Algorithm}) on random graphs as well as on real-world graphs from the 10th~DIMACS challenge~\cite{Dim12}. 
Our implementation solves the optimization versions of the problems and finds $t$-well-connected 2-clubs of maximum size.
We refer to our implementation in the following as \texttt{WCC} (well-connected 2-club). 
For the special case of solving the basic \textsc{2-Club} problem (see \cref{ssec:2-club}), we compare \texttt{WCC} against the programs of \citet{hartung_experiments}\footnote{Available from~\url{http://fpt.akt.tu-berlin.de/two_club/}.} (referred to as \texttt{HKN}) and \mbox{\citet{chang}}\footnote{Available from~\url{https://sites.google.com/site/kdynamicneighborhoodproject/}.} (referred to as \texttt{CHLS}).
Furthermore, in \cref{ssec:biconnected-club} we also compare against two programs of \citet{YPB17}\footnote{\citet{YPB17} sent us the code of their two programs.} (one branch-and-bound algorithm, one ILP) for finding $2$-connected~$2$-clubs.
In \cref{ssec:2-t-club} we analyze the performance of \texttt{WCC} for finding $t$-well-connected 2-clubs in social networks for~$t \ge 1$.
In all our experiments, we set the time limit to one hour. 
All time measurements include the time to read the input graph.

We ran all our experiments on an Intel(R) Xeon(R) CPU E5-1620 3.60\,GHz machine with 64\,GB main memory under the Debian GNU/Linux 7.0 operating system. 
Our implementation is in Java and runs under the OpenJDK runtime environment in version~1.7.0\_65.  
The C++ implementation of \citet{YPB17} was compiled with g++ version 5.4.0.

\subsection{Comparison with other \textsc{2-Club} Algorithms} \label{ssec:2-club}

We may use our implementation (\texttt{WCC}) in three ways to find 2-clubs: by looking for 1-robust 2-clubs, 0-hereditary 2-clubs, or 1-connected 2-clubs.
We report the running times for the setting where we look for $0$-hereditary 2-clubs since this is faster than the other two settings. 
The setting where we look for $1$-connected 2-clubs is particularly slow since our implementation uses a maximum flow computation to verify that the graph is connected.

\paragraph*{Random Graphs}
As in previous experimental evaluations~\cite{hartung_experiments,chang}, we use the random graph generator due to \citet{gendreau} where the density of the resulting graphs is controlled by two parameters, $0\le a\le b\le 1$, and the expected density is $(a+b)/2$.

The results on random graphs give a clear picture: 
As our algorithm is not designed for random graphs but large sparse real-world networks, the implementations of \citet{hartung_experiments} (\texttt{HKN}) and \citet{chang} (\texttt{CHLS}) are both significantly faster than ours.
However, \texttt{WCC} still outperforms our \texttt{ILP}.
\Cref{tab:random-2-club-summary} summarizes the results for random graphs with expected density~0.15, which produces the hardest instances as already observed earlier~\cite{bourjolly_exact,hartung_experiments,chang}.

\begin{table}
	\caption{\normalfont
		Results on random graphs of density 0.15. 
		All time measurements are the average over 100 instances with the described parameters. 
		All running times are given in seconds.
		If a running time is bold, then some of these 100 instances (the number is given in brackets) did not finish within the time limit of one hour and the respective instance is accounted with one hour in the average. 
		Bold values are therefore just lower bounds on the real running times.%
	}
	\footnotesize
	\label{tab:random-2-club-summary}
	\renewcommand{\tabcolsep}{5pt}
	\centering
	\begin{tabular}{rrrrrrr}
		\toprule
		$a$ & $b$ 	& $n$ & \texttt{WCC} & \texttt{HKN} & \texttt{CHLS} & \texttt{ILP} \\ \midrule
		0.0 & 0.3 	& 140 & 3.50 & 1.27 & 1.00 & 3.61 \\
			&  		& 150 & 5.71 & 1.71 & 1.82 & 4.48 \\
			&  		& 160 & 6.84 & 1.99 & 2.64 & 5.03 \\
			\midrule
		0.05& 0.25 	& 140 & 19.34 & 6.11 & 7.37 & 38.67 \\
			&  		& 150 & 39.54 & 11.98 & 18.00 & 63.23 \\
			&  		& 160 & 72.35 & 18.33 & 30.54 & 95.04 \\
			\midrule
		0.1 & 0.2 	& 140 & 70.09 & 24.34 & 29.99 & 246.20 \\
			&  		& 150 & 183.04 & 59.41 & 85.94 & \textbf{605.76 [2]} \\
			&  		& 160 & 557.92 & 178.16 & 306.51 & \textbf{1,115.38 [9]} \\
			\midrule
		0.15& 0.15 	& 140 & 105.68 & 39.85 & 48.62 & 612.50 \\
			&  		& 150 & 322.70 & 118.77 & 168.91 & \textbf{1,508.00 [8]} \\
			&  		& 160 & \textbf{1,132.05 [1]} & 384.38 & \textbf{611.65 [1]} & \textbf{2,594.98 [46]} \\
		\bottomrule
	\end{tabular}
\end{table}

The algorithm \texttt{HKN} could solve all random instances within the time limit of one hour.
Our algorithm \texttt{WCC} as well as the algorithm \texttt{CHLS} could solve all but one instance within this time limit.
Both \texttt{WCC} and \texttt{CHLS} fail to solve the same instance, which \texttt{HKN} could solve in 2,571 seconds.
The \texttt{ILP}, however, failed on 65 of the overall 1,200 instances.

On average, \texttt{WCC} is 3.1 times slower than \texttt{HKN} and 2.4 times slower than \texttt{CHLS}.
If one considers the accumulated running times over all instances, then \texttt{WCC} is 3 times slower than \texttt{HKN} and 1.9 times slower than \texttt{CHLS}.
Compared to the \texttt{ILP} on instances where it found a solution within the time limit, \texttt{WCC} is on average 2.4 times faster and considering the accumulated running times (with unsolved instances counted with one hour), \texttt{WCC} is at least 2.7 times faster.

The reason for this relatively weak performance of \texttt{WCC} is the extensive usage of data reduction rules. 
The algorithm tries to apply all data reduction rules, but on random data they seldom apply.
To improve the performance on random graphs, improving the efficiency of the data reduction rules as well as investigating heuristic selection strategies concerning when to apply which rule could be a promising approach for improvements. 

\paragraph*{Real-World Graphs}
We considered real-world graphs from the 10th DIMACS challenge~\cite{Dim12}. 
We ran our algorithm on instances from the categories \emph{Clustering} and \emph{Walshaw's Graph Partitioning Archive}~\cite{Dim12}; this is a standard benchmark for graph clustering and community detection algorithms also used by \citet{YPB17} to evaluate their implementation for finding biconnected 2-clubs.
To test our algorithm on large scale social network graphs we ran it also on graphs from the \emph{co-author and citation} category~\cite{Dim12} and a graph (\textsf{graph-thres-01}) mined from DBLP\footnote{\url{http://dblp.uni-trier.de/faq/Am+I+allowed+to+crawl+the+dblp+website.html}} in 2012.
\pgfplotstableset{
	create on use/avgDeg/.style={
		create col/expr={2*\thisrow{m}/\thisrow{n}}}
}%
\pgfplotstableset{
	create on use/NormDens/.style={
		create col/expr={2*(\thisrow{m}/\thisrow{n})/(\thisrow{n}-1)}},
}%
\begin{table}[t!]
	\caption{
		Our data set sorted by the number of edges. 
		Here, $n^*$ is the number of vertices (excluding isolated vertices that are discarded when reading the graph), $m$ is the number of edges, density is defined as $m/\binom{n^*}{2}$, $\Delta$ is the maximum degree, and size is the number of vertices in a largest 2-club.
	}
	\label{tab:graphs}
	\footnotesize\centering
\pgfplotstabletypeset[columns={file,n,m,NormDens,avgDeg,maxDeg,2-club},
	columns/file/.style={string type,column name=Graph,column type = {r}},
	columns/n/.style={column name=$n^*$,precision=1,column type = {r}},
	columns/m/.style={column name=$m$,precision=1,column type = {r}},
	columns/avgDeg/.style={column name=avg.\ deg.,precision=2,dec sep align},
	columns/NormDens/.style={column name=density,precision=2,dec sep align},
	columns/maxDeg/.style={column name=$\Delta$,precision=1,column type = {r}},
	columns/2-club/.style={column name=size,precision=1,column type = {r}},
	every head row/.style ={before row=\toprule, after row=\midrule},
    every last row/.style ={after row=\bottomrule}]{\dataGraphs}
\end{table}%
\Cref{tab:graphs} shows our test set of graphs and \Cref{tab:real-world-club-summary} overviews the results for the  real-world graphs.
\begin{table}[t!]
	\centering
	\caption{\normalfont
		Results for \textsc{2-Club} on real-world graphs with a time limit of one hour per instance. 
		Here, $n^*$~is the number of vertices with degree at least one,~$m$~is the number of edges, and size is the number of vertices in a largest 2-club.
		Empty cells represent timeouts (for \texttt{WCC} or \texttt{HKN}) or insufficient memory (for \texttt{CHLS}).
		All times are in seconds. 
	}
	\footnotesize
	\label{tab:real-world-club-summary}
\pgfplotstabletypeset[columns={file,n,m,sizeMarten,timeMarten,timeSepp,timeChan},
	columns/file/.style={string type,column name=Graph,column type = {r}},
	columns/n/.style={column name=$n^*$,precision=1,column type = {r}},
	columns/m/.style={column name=$m$,precision=1,column type = {r}},
	columns/sizeMarten/.style={column name=size,precision=1,column type = {r}},
	columns/timeMarten/.style={column name=\texttt{WCC},precision=2,dec sep align},
	columns/timeSepp/.style={column name=\texttt{HKN},precision=2,fixed,dec sep align},
	columns/timeChan/.style={column name=\texttt{CHLS},precision=2,fixed,dec sep align},
	every head row/.style ={before row=\toprule, after row=\midrule},
    every last row/.style ={after row=\bottomrule}]{\dataTwoClubs}
\end{table}

Similarly to the experiments on random graphs, \texttt{CHLS} was the fastest solver on small instances.
This is probably due to the fact that \texttt{CHLS} is written in C (\texttt{WCC} and \texttt{HKN} are written in Java) and it uses an adjacency matrix as graph representation (\texttt{WCC} and \texttt{HKN} use adjacency lists).
Due to the usage of an adjacency matrix to store the graph, \texttt{CHLS} could not store and thus not solve any large graph in the \emph{Co-author} category.

\texttt{WCC} could solve all but one of the instances of the smaller graphs (less than 50,000 edges) within one second; the exception, namely \textsf{polblogs}, needed less than four seconds.
On medium-size and large instances, \texttt{WCC} clearly outperformed the other algorithms.
The only instance that \texttt{WCC} could not solve within one hour is the largest graph \textsf{coPapersCiteseer}.
A second run without time limit shows a running time of 1.8~hours for this instance; the second-best solver \texttt{HKN} needs 8.6 \emph{days}.
This good performance is mainly due to the extensive usage of data reduction rules---a feature that slows down our algorithm on random graphs.

\subsection{Comparison with other \textsc{$2$-Connected 2-Club} Algorithms} \label{ssec:biconnected-club}
We considered the same real-world instances as in \cref{ssec:2-club} to compare against the two implementations of \citet{YPB17}.
Their first implementation, which we refer to as \texttt{BB}, is a combinatorial branch-and-bound algorithm with some lower bound heuristics. 
Their second implementation, which we refer to as \texttt{BC}, is a branch-and-cut algorithm that is based on an ILP formulation. 
The results of the three algorithms are shown in \cref{tab:biconnected}.
\begin{table}[t!]
	\centering
	\caption{\normalfont
		Results for \textsc{2-Connected 2-Club} on real-world graphs with a time limit of one hour per instance. 		
		Here, $n^*$~is the number of vertices with degree at least one,~$m$~is the number of edges, and size is the number of vertices in a largest 2-club.
		Empty cells represent timeouts.
		All times are in seconds. 
	}
	\footnotesize
	\label{tab:biconnected}
\pgfplotstabletypeset[columns={Datafile,n,m,size (our),time (our),time (BB),time (ILP)},
	columns/Datafile/.style={string type,column name=Graph,column type = {r}},
	columns/n/.style={column name=$n^*$,precision=1,column type = {r}},
	columns/m/.style={column name=$m$,precision=1,column type = {r}},
	columns/size (our)/.style={column name=size,precision=1,column type = {r}},
	columns/time (our)/.style={column name=\texttt{WCC},precision=2,fixed,dec sep align},
	columns/time (BB)/.style={column name=\texttt{BB},precision=2,fixed,dec sep align},
	columns/time (ILP)/.style={column name=\texttt{BC},precision=2,fixed,dec sep align},
	every head row/.style ={before row=\toprule, after row=\midrule},
    every last row/.style ={after row=\bottomrule}]{\dataBiconTwoClubs}
\end{table}

\texttt{WCC} solved 21 out of the 26 instances within the time limit of one hour 
and all but one of the smaller graphs with less than 50,000 edges within less than one minute.
Only one instance, namely \textsf{polblogs}, could not be solved within one hour.
Moreover, \texttt{WCC} solved two larger graphs with more than 800,000 edges within less than 10 minutes.
As already discussed in \Cref{sec:Algorithm}, \texttt{WCC}  is not as efficient for the $t$-connected 2-club model as for the other two models.
In fact, a closer inspection of our implementation on the graph \textsf{polblogs} showed that more than 95\% of the running time went into the maximum flow-based connectivity check. 
Herein, the reduction to the maximum flow instance, without the computation of the maximum flow itself, took more than half of the time.
Although this reduction can be performed in linear time, it is applied very often after branching or after applying the data reduction rules. 

\texttt{BB} solved 11 and \texttt{BC} solved 13 out of the 26 instances.
Both of these solvers were faster on the smallest instances and on a few medium size instances.
In fact, the smaller graph \textsf{polblogs} that \texttt{WCC} could not solve is solved by \texttt{BC} in around eight minutes.
However, \texttt{WCC} solved three smaller graphs (\textsf{power}, \textsf{uk}, \textsf{add32}) with less than 10,000 edges in less than a second each and both \texttt{BB} and \texttt{BC} could not solve them in less than half an hour.
Neither \texttt{BB} nor \texttt{BC} could solve any of the larger graphs with more than 800,000 edges.
In fact, the ILP-based solver \texttt{BC} did not even finish writing all constraints within one hour for these large graphs.

Summarizing, \texttt{WCC} solved substantially more instances than \texttt{BB} and \texttt{BC}.
Again, the good performance is mainly due to the extensive usage of data reduction rules.
Furthermore, as one can see in the results for \textsf{polblogs}, if the data reductions are not effective enough, then \texttt{WCC} needs much more time to solve the instance.

We did not perform tests on random graphs since we expect that the outcome is similar as in \cref{ssec:2-club}: 
\texttt{WCC} will be outperformed by \texttt{BB}~and~\texttt{BC}.

\subsection{Evaluation for~$t \ge 1$} \label{ssec:2-t-club}

Finally, we performed experiments for computing $t$-well-connected 2-clubs with~$t \ge 1$. 
For these experiments, we compare the running times of \texttt{WCC} and \texttt{ILP}. 
Moreover, we compare our three well-connected 2-club models with each other. 
Herein, we compare $t$-robust 2-clubs and $t$-connected 2-clubs with $(t-1)$-hereditary 2-clubs.
The reason is that the base case---2-clubs---is reached for different~$t$-values: 
A 2-club is a $1$-robust 2-club, a $1$-connected 2-club, and a $0$-hereditary 2-club.
Furthermore, one can see many similar solutions when comparing $t$-robust 2-clubs against $(t-1)$-hereditary 2-clubs.
We ran our implementation with values for~$t$ being 1, 2, 3, 4, 5, 7, 9, 10, 15, 20, 50, 100, and 1,000.
The performance of our implementation varies for different values of~$t$ but there seems to be no clear correlation between~$t$ and the performance.

In our tests \texttt{WCC} clearly outperformed \texttt{ILP}.
Furthermore, \texttt{ILP} could solve only small real-world graphs.
Due to this one-sided result, we only discuss the results of \texttt{WCC}.
We refer to \mbox{\Cref{sec:appendix}} for a complete list of the experiments on real-world graphs (including \texttt{ILP}).
Comparing \texttt{WCC} on the three models, unsurprisingly, the results for $(t-1)$-hereditary 2-clubs and $t$-robust 2-clubs are very similar (see \cref{fig:time-diff-models}). 
\begin{figure}[t!]
	\centering
	\runtimeDiagram{add20}
	\runtimeDiagram{coAuthorsCiteseer}
	\runtimeDiagram{coAuthorsDBLP}
	\runtimeDiagram{coPapersCiteseer}
	\caption{
		Running time plots for finding $t$-well-connected 2-clubs for different values of~$t$ on \textsf{add20} (top left), \textsf{coAuthorsCiteseer} (top right), \textsf{coAuthorsDBLP} (bottom left), and \textsf{coPapersCiteseer} (bottom right).
		All running times are in seconds.
		In all plots, black circles correspond to $t$-connected 2-clubs, red triangles correspond to $(t-1)$-hereditary 2-clubs, and blue stars correspond to~$t$-robust 2-clubs.
		In the two plots on top, the instances could be solved for all~$t$ and all three models.
		In the bottom-left plot, one black marker (for~$t=7$) is missing due a time out.
		In the bottom-right plot several markers are missing due to time outs.
		Overall, smaller values of~$t$ seem harder. 
		However, there is no clear trend showing which concrete~$t$-values lead to the computationally hardest instances.
	}
	\label{fig:time-diff-models}
\end{figure}
In both variants, \texttt{WCC} solved the same 329 of the overall 338 instances within one hour.
On average, the $t$-robust 2-club could be computed 5\% faster.
Probably, this is due to the fact that for computing $(t-1)$-hereditary 2-clubs we run a clique-algorithm if the $(t-1)$-hereditary 2-club has size at most~$t$ (see \cref{smallClubs}).
Only in 20 of the 175 instances that contain a $t$-robust 2-club the maximum $(t-1)$-hereditary 2-club is larger than the maximum $t$-robust 2-club (see also \cref{fig:size-diff-models}).
\begin{figure}[t!]
	\centering
	\sizeDiagram{add20}
	\sizeDiagram{coAuthorsCiteseer}
	\sizeDiagram{coAuthorsDBLP}
	\sizeDiagram{coPapersCiteseer}
	\caption{
		The ordering of the respective maximum $t$-well-connected 2-clubs for different values of~$t$ on the same graphs: \textsf{add20} (top left), \textsf{coAuthorsCiteseer} (top right), \textsf{coAuthorsDBLP} (bottom left), and \textsf{coPapersCiteseer} (bottom right).
		In all plots, black circles correspond to $t$-connected 2-clubs, red triangles correspond to $(t-1)$-hereditary 2-clubs, and blue stars correspond to~$t$-robust 2-clubs. 
		The instances are the same as in \Cref{fig:time-diff-models} and the same markers are missing due to time-outs.
		Again, the differences between $(t-1)$-hereditary 2-clubs and $t$-robust 2-clubs are minimal. 
		Only if there is no $t$-robust 2-club and the largest $(t-1)$-hereditary 2-club is a clique, then there is a visible difference.
		The $t$-connected 2-club can be significantly larger than the other 2-clubs (see the two plots on the left side).
	}
	\label{fig:size-diff-models}
\end{figure}
The largest relative difference in sizes (in cases where a $t$-robust 2-club exists) was observed in the graph \textsf{adjnoun} where the largest $4$-robust 2-club has size six and the largest $3$-hereditary 2-club has size nine.

For \textsf{citationCiteseer}, \textsf{coAuthorsCiteseer}, \textsf{coAuthorsDBLP}, and \textsf{graph-thres-01}, \texttt{WCC} finds the largest $(t-1)$-hereditary 2-club and the largest $t$-robust 2-club relatively quickly (always within three minutes). This is not the case for the two largest graphs \textsf{coPapersCiteseer} and \textsf{coPapersDBLP}.
While \textsf{coPapersCiteseer} could be solved within one hour for all values of~$t$ except 1 and 3, \textsf{coPapersDBLP} could not be solved for~$t$ ranging from 2 to 10. 
However, experiments without time limit and with further statistics collected during all stages of the implementation solved all instances within two hours. 
Since computing and outputting these additional data affected the running time quite a lot, we did not include the measured running times in the appendix. 

\texttt{WCC} found the largest $t$-connected 2-club on many instances (290 of the 338 instances were solved).
While only one of the smaller graphs (\textsf{polblogs}) caused trouble for \texttt{WCC}, \texttt{WCC} struggled on large graphs with more than 800,000 edges; only 39 of the 78 instances corresponding to these large graphs could be solved within the time limit of one hour.
On the instances \texttt{WCC} could deal with it needed much more time than for finding the largest $t$-robust or $(t-1)$-hereditary 2-club (see also \cref{fig:time-diff-models}).
A good example here is the graph \textsf{coAuthorsDBLP} where finding the largest $9$-connected 2-club was more than 100 times slower than finding the largest $8$-hereditary 2-club.

\paragraph*{Detailed Running Time Analysis for $(t-1)$-Hereditary and $t$-Robust 2-Club Model} 
A closer inspection of the \textsf{coPapersDBLP} graph shows that, for all~$t > 1$, our implementation spent around 75\% of the time for constructing the Turing kernels and initializing the data structures (common neighbor matrix, incompatibility graph, and compatibility vector; see \cref{sec:Algorithm}). 
The rest of the time is spent mostly on applying the data reduction rules; for each value of~$t$, less than 30 seconds were spent on the actual branching. 
For~$t=1$, the maximum degree of 3,299 gives a lower bound of 3,300 for the solution, which in this case is the solution size.
Thus, all but six Turing kernels can be dismissed due to too small size before initializing the data structures (see discussion in \cref{ssec:data-structures}).
This explains why the implementation was that much faster for~$t=1$.

The \textsf{coPapersCiteseer} graph is quite different.
This can be seen for example in the solution size: For increasing~$t$ the size of the solution decreases only slightly (see \cref{fig:size-diff-models} bottom-right plot); this is in stark contrast to the \textsf{coPapersDBLP} graph.
Furthermore, for~$t=1$ on the \textsf{coPapersCiteseer} graph, there were 3,175 Turing kernels that were all larger than the lower bound.
The average number of vertices in these Turing kernels was 1,389; the maximum was 2,022, the minimum 1,189. 
Recall that our data structures involve square matrices (the common neighbor matrix) whose number of rows and columns is equal to the number of vertices in the respective Turing kernel.
This explains why the initialization of these data structures was the bottleneck in our implementation.

\subsection{Overall Conclusions from Experiments}
We demonstrated the effectiveness of our general search-tree approach to find
$t$-well-connected 2-clubs in large sparse real-world networks.  
In particular, to the best of our knowledge, we provide the first results for large sparse social networks with more than a million edges.  
The key ingredients for our implementation are the extensive usage of Turing kernels and data reduction rules. 

For finding $t$-connected 2-clubs there is room for improving our implementation.
To this end, the most obvious approach is to look for an efficient data structure that can quickly delete vertices and answer connectivity queries.
Finding good lower bounds is another approach to speed up our implementation as we only use a lower bound for the base case of finding 2-clubs.

Our experiments also indicated that the common ILP formulations are not able to cope with the large graphs we consider.
The reason is simply that creating all constraints is too time consuming.\footnote{To the best of our knowledge, ILP formulations for \textsc{2-Club} and its variants have at least~$n^2-m$ constraints~\cite{BS17,YPB17,VB12}. Even for the smallest graph from the \emph{co-author and citation} category~\cite{Dim12}, this amounts to more than $5.2 \cdot 10^{10}$ constraints which, on the hardware we used, could not be generated within one hour.} 
Combining Turing-kernelization and data reduction rules with ILPs using row generation and other sophisticated tricks might lead to a competitive program.

\section{Outlook}
\label{sec:outlook}
We studied three established models for cohesive subgraphs, namely $t$-robust, $t$-hereditary, and $t$-connected 2-clubs.
These are considered to overcome the typical hub-and-spoke structure in the solutions of classical 2-\textsc{Club} problem.
We presented theoretical findings for these models which provided the basis for our implementation of an exact algorithm for the problems.
Our experiments demonstrate the efficiency of the presented algorithm on large sparse real-world graphs.

We conclude with some challenges for future research:
\begin{itemize}
\item Our algorithmic result with respect to maximum degree (see \Cref{maxdegthm}) does not fully explain the success of our implementation on our data set as the parameter value is too large (see \cref{tab:graphs}). 
	Moreover, we have NP-hardness even if the smaller parameters average degree and degeneracy are constant (see \Cref{structuralParam}).
	Thus, finding the ``right'' structure that explains the practically observed performance also in theory remains a task for future research.

\item
	It is open to study the polynomial-time approximability of the different $t$-well connected 2-club variants.
	Notably, \textsc{2-Club} has a factor-$O(n^{1/2})$ approximation algorithm and is NP-hard to approximate within a factor of~$O(n^{1/2-\varepsilon})$~\cite{asahiro}.
\item 
	An extension of the algorithmic approach to 3-clubs and beyond is another challenge. 
	\citet{AC14} presented an ILP formulation for $t$-robust $3$-clubs and 
	\citet{VB12} presented an ILP formulation for a relaxation of $t$-robust $s$-clubs with~$s>2$.
	Solving the relaxed problems turned out to be computationally very hard; we are not aware of experimental evaluations of the exact formulation for $t$-robust $3$-clubs.  
\item
	\citet{VPBP14} studied a variation of $t$-robust 2-clubs (called $\gamma$-relative-robust 2-clubs) where essentially the constant~$t$ is replaced by a function in the size of the club.
	Our algorithms do not work for this variant since we always assume that at most one of two incompatible vertices (see \cref{def:compatible} in \cref{sec:properties}) can be in a $t$-well-connected 2-club.
	This assumption does not hold for $\gamma$-relative-robust 2-clubs.
	Extending our work to this model remains a task for future work.
\item 
	The combination of efficient data reduction rules with mathematical programming techniques may lead to further accelerations in finding $t$-well-connected 2-clubs.
\end{itemize}

\paragraph*{Acknowledgment}
We are grateful to four anonymous reviewers of \emph{EJOR} for their careful and constructive feedback.
We thank Oleksandra Yezerska, Foad Mahdavi Pajouh and Sergiy Butenko \cite{YPB17} for providing us with their source code of their programs.

\bibliographystyle{abbrvnat} 
\bibliography{robust-two-club}

\newpage

\appendix
\section{Full Experimental Results for~$t \ge 1$}\label{sec:appendix}

\pgfplotstableset{
	row sep=\\,
	begin table=\begin{longtable},
	end table=\end{longtable},
	every head row/.append style={before row=\caption{Results for the real-world graphs. The columns time and size denote the running time and the size for the respective model: $t$-robust (r), $(t-1)$-hereditary (h \& ILP; for \texttt{WCC} and our ILP), and $t$-connected (c) 2-clubs.}\\\toprule, after row=\midrule\endhead},
}

\footnotesize\centering
\pgfplotstabletypeset[columns={file,t,sizeRobust,timeRobust,sizeHereditary,timeHereditary,timeILP,sizeConnected,timeConnected},
columns/file/.style={
	column name={graph},
	assign cell content/.code={
		\pgfmathparse{mod(\pgfplotstablerow,13)}
		\pgfmathtruncatemacro{\myint}{ \pgfmathresult}  
		\ifnum\myint=0
			\pgfkeyssetvalue{/pgfplots/table/@cell content}
			{\multirow{13}{*}{\rotatebox[origin=c]{90}{##1}}}%
		\else
			\pgfkeyssetvalue{/pgfplots/table/@cell content}{\nopagebreak}%
		\fi
	},
},
columns/t/.style={column name=$t$,precision=1,column type = {r}},
columns/timeHereditary/.style={column name=time (h),fixed,precision=2,dec sep align},
columns/sizeHereditary/.style={column name=size (h),fixed,precision=1,dec sep align},
columns/timeILP/.style={column name=time (ILP),fixed,precision=2,dec sep align},
columns/timeConnected/.style={column name=time (c),fixed,precision=2,dec sep align},
columns/sizeConnected/.style={column name=size (c),fixed,precision=1,dec sep align},
columns/timeRobust/.style={column name=time (r),fixed,precision=2,dec sep align},
columns/sizeRobust/.style={column name=size (r),fixed,precision=1,dec sep align},
every nth row={13}{before row=\nopagebreak\midrule},
every last row/.style ={after row=\bottomrule}]{\dataAllModels}

\end{document}